\documentclass[aip, jmp, nofootinbib]{revtex4-1}

\usepackage{amsfonts, amsmath, amssymb}
\usepackage{amsthm, bm}
\usepackage{comment}

\theoremstyle{definition}
\newtheorem{Def}{Definition}
\newtheorem{Rem}[Def]{Remark}

\theoremstyle{plain}
\newtheorem{Thm}[Def]{Theorem}
\newtheorem{Lem}[Def]{Lemma}
\newtheorem{Prop}[Def]{Proposition}
\newtheorem{Cor}[Def]{Corollary}

\newcommand{\cpl}[2]{\langle\langle #1, #2 \rangle\rangle}
\newcommand{\snabla}{\nabla\hspace{-0.67em}/\hspace{0.2em}}

\begin{document}

\sloppy

\title{Lorentzian positive mass theorem for spacetimes with distributional curvature}

\author{Keigo Shibuya}
\affiliation{Graduate School of Mathematics, Nagoya University, Nagoya 464-8602, Japan}


\begin{abstract}
In this paper, we prove Lorentzian positive mass theorem for spacetimes with distributional curvature.
To do so, we introduce distributional curvature and generalized Arnowitt-Deser-Misner (ADM) momentum.
As an application, we discuss a junction of spacetimes.
\end{abstract}


\maketitle

\section{Introduction \label{sec:intro}}

In general relativity, the positive mass theorem holds for asymptotically flat spacetimes, which guarantees the stability of spacetimes.
There are two types, Riemannian \cite{SY1} and Lorentzian version \cite{SY2,W,PT}.
The former version is for time-symmetric initial data with non-negative Ricci scalar.
The latter has been proven for spacetimes satisfying the dominant energy condition.
The theorem is usually proven under the assumption of enough smoothness.
In more detail, it is assumed that the metric has regular second derivatives to define its Riemannian curvature locally (for example, see Refs. \onlinecite{D} and \onlinecite{GT}).
Meanwhile, Miao showed that Riemannian positive mass theorem admits a jump of extrinsic curvature \cite{M}.
For instance, it guarantees the stability of spacetimes with a thin-shell matter.
Recently, Lee and LeFloch proved Riemannian positive mass theorem for spacetimes with distributional curvature \cite{LL}.

In this paper, employing Witten type proof based on spinor \cite{W}, we extend Lee \& LeFloch's work to Lorentzian cases.
At the same time, we will make a minor correction for their work \footnote{Lee \& LeFloch assumed $W^{1,n}_{-\tau}$-asymptotic flatness for $\tau\geq\tau_0=\frac{n-2}{2}$ or $\tau=\tau_0$ through Ref. \onlinecite{LL}. But, some proofs therein failed in equality case $\tau=\tau_0$ and then the corresponding results require each correction. So, in this paper, we give corrected statements.}.

The main theorem of this paper is summarized as

\begin{Thm}[Lorentzian positive mass theorem with distributional curvature]\label{thm:main}
Let $(\Sigma,g_{ij},K_{ij})$ be an $n$-dimensional $W^{1,n}_{-\tau}$-asymptotically flat data with $\tau>\tau_0:=\frac{n-2}{2}$ (definition \ref{def:af}) and $P$ its generalized ADM $(n+1)$-momentum (definition \ref{def:gam}).
If this data has a spin structure and satisfies the dominant energy condition in distributional sense (definition \ref{def:dec}), then $P(U)$ is non-negative for any future-directed vector $U$ which is constant on the frame $\Phi$.
In addition, if $P$ is zero, then the data has a globally parallel spinor frame with respect to the spacetime connection $\nabla$.
\end{Thm}

The existence of a globally parallel spinor field means that spacetime is flat. 

The organization of this paper is as following:
at first, we give a definition of distributional curvature in Sec. \ref{sec:dc}.
Next, in Sec. \ref{sec:af-gam}, we introduce asymptotically flat data and generalize the notion of ADM momentum.
In Sec. \ref{sec:lwf}, we prove Lichnerowicz-Weitzenb{\"o}ck formula for distributional curvature.
In Sec. \ref{sec:mt}, we give the main theorem and its proof.
In Sec. \ref{sec:ex}, we consider a data with corner as an example.

\section{Distributional curvature \label{sec:dc}}

In this section, we define distributional curvature.

Let $\Sigma$ be an $n$-dimensional smooth manifold and take an auxiliary Riemannian metric $h_{ij}$ on $\Sigma$.
We remark that the following arguments in this section are independent of the choice $h_{ij}$.
Using its volume measure $dh$ and Levi-Civita connection $\underline{D}_i$, we define the Lebesgue space $L^p_\mathrm{loc}$ and the Sobolev space $W^{k,p}_\text{loc}$ of functions or tensor fields on $\Sigma$.

At first, we look at a smooth data $(\Sigma, g_{ij}, K_{ij})$, where $g_{ij}$ is a Riemannian metric and $K_{ij}$ is a symmetric tensor.
We regard the data as a hypersurface with future-directed normal $t^\mu$ in $(n+1)$-dimensional spacetime $(\mathcal{M},g_{\mu\nu})$.
Then, $g_{ij}$ is the $n$-dimensional component of the induced metric $q_{\mu\nu}:=g_{\mu\nu}+t_\mu t_\nu$ from $g_{\mu\nu}$ and $K_{ij}$ is the extrinsic curvature $K_{\mu\nu}=q_\mu^{\ \alpha}\nabla_\alpha t_\nu$, where $\nabla$ denotes the Levi-Civita connection of $g_{\mu\nu}$.
As above, we use Latin and Greek indices for $n$-dimensional and $(n+1)$-dimensional components of tensors, respectively.
Now, let us suppose that the Einstein equation
\begin{equation}
\mathcal{R}_{\mu\nu}-\frac{1}{2}\mathcal{R}g_{\mu\nu}=T_{\mu\nu}
\end{equation}
holds, where $\mathcal{R}_{\mu\nu}$ and $T_{\mu\nu}$ denote the Ricci tensor of spacetime manifold and stress tensor of matters, respectively.
By using Gauss-Codazzi equations, a part of the Einstein equation gives us the constraint equations
\begin{subequations}
\label{eq:ce}
\begin{eqnarray}
\frac{1}{2}(R-K_{\mu\nu}K^{\mu\nu}+K^2) &=& \mu,\\
D_\beta(K^\beta_{\ \alpha} - K q^\beta_{\ \alpha}) &=& J_\alpha,
\end{eqnarray}
\end{subequations}
where $R$ and $D_\alpha$ are the Ricci scalar and Levi-Civita connection of $g_{ij}$.
Here $\mu:=T_{\mu\nu}t^\mu t^\nu$ and $J_\alpha:=T_{\mu\nu}t^\mu q_\alpha^{\ \nu}$ correspond to the energy and momentum densities of matters.
For later convenience, we write the left-hand side of Eqs. (\ref{eq:ce}) by $H_G$ and $M_G$, that is,
\begin{subequations}\label{eq:def-HM}
\begin{eqnarray}
\label{eq:def-H}
H_G &:=& \frac{1}{2}(R-K_{ij}K^{ij}+K^2),\\
\label{eq:def-M}
(M_G)_i &:=& D_j({K^j}_i - K \delta^j_{\ i}).
\end{eqnarray}
\end{subequations}

Next, to define distributional curvature, it is nice to divide $H_G$ and $M_G$ into two parts, that is, a part consisting of the second derivatives of $g_{ij}$ or the first derivatives of $K_{ij}$, and others.
$H_G$ includes the second derivatives of $g_{ij}$ in the Ricci scalar $R$ as
\begin{equation}
R = \underline{D}_iV^i + F,
\end{equation}
where
\begin{subequations}
\begin{equation}
\label{eq:def-V}
V^i := g^{ij}g^{kl}(\underline{D}_kg_{lj}-\underline{D}_jg_{kl})
\end{equation}
and
\begin{equation}
\label{eq:def-F}
F := g^{ij} \underline{R}_{ij}-\Gamma_{ij}^k\underline{D}_kg^{ij}+\Gamma_{ji}^j\underline{D}_kg^{ik} +g^{ij}(\Gamma_{kl}^k\Gamma_{ij}^l-\Gamma_{jl}^k\Gamma_{ik}^l).
\end{equation}
\end{subequations}
Here $\Gamma_{ij}^k=g^{kl}(\underline{D}_ig_{jl}+\underline{D}_jg_{il}-\underline{D}_lg_{ij})$ is the difference of two connections and $\underline{R}_{ij}$ is the Ricci tensor with respect to $h_{ij}$.
Similarly, for $M_G$, we have
\begin{equation}
(M_G)_i = \underline{D}_j({K^j}_i-K{\delta^j}_i)+W_{ij}^k({K^j}_k-K{\delta^j}_k),
\end{equation}
where $W_{ij}^k:=\Gamma_{jl}^l\delta^k_i-\Gamma_{ij}^k$.
These expressions motivate us to define distributional curvature as below.

\begin{Def}[Distributional Curvature]
Assume $g_{ij} \in L^\infty_\text{loc}\cap W^{1,2}_\text{loc}$ and $K_{ij}\in L^2_\text{loc}$.
Then we can define the distributional curvature $(H_G, M_G)$ by integrals
\begin{subequations}\label{eq:def-dc}
\begin{equation}\label{eq:def-dc-H}
\cpl{H_G}{u}  = \int_\Sigma\frac{1}{2}\left(-V^i\underline{D}_i\left(u\frac{dg}{dh}\right)+(F+K_{ij}K^{ij}-K^2)u\frac{dg}{dh}\right)dh
\end{equation}
and
\begin{equation}\label{eq:def-dc-M}
\cpl{M_G}{v} = \int_\Sigma\left(-({K^j}_i-K{\delta^j}_i)\underline{D}_j\left(v^i\frac{dg}{dh}\right)+W_{ij}^k({K^j}_k-K{\delta^j}_k)v^i\frac{dg}{dh}\right)dh
\end{equation}
\end{subequations}
for all smooth functions $u$ and vectors $v^i$ with compact supports.
In this equations, $dh$ and $dg$ are respectively the volume measures of $h_{ij}$ and $g_{ij}$, and $\frac{dg}{dh}$ is the Radon-Nikodym derivative in the measure theory.
\end{Def}

\begin{Rem}
By $g_{ij}\in L^\infty_\text{loc}$ in the above, we mean not only $g_{ij}\in L^\infty_\text{loc}$ but also $g^{ij}\in L^\infty_\text{loc}$.
This condition ensures $g_{ij}\in W^{k,p}_\text{loc} \Leftrightarrow g^{ij}\in W^{k,p}_\text{loc}$.
Similarly, by $g_{ij}\in \mathcal{B}^0$, we mean $g_{ij} \in \mathcal{B}^0$ and $g^{ij} \in \mathcal{B}^0$, where $\mathcal{B}^0$ denotes the space of bounded continuous fields.
\end{Rem}

In the following, we simply call such pair $U=(u, v^i)$ a vector and say that the vector is future-directed if $u\geq\sqrt{g_{ij}v^iv^j}$.

Next, we introduce the dominant energy condition \cite{HE,Wa}.
In smooth cases, it is expressed as $\mu\geq\sqrt{g^{ij}J_iJ_j}$.
In our settings, however, the quantities $\mu$ and $J_i$ are also distributions through the Einstein equation.
So the dominant energy condition for them is defined as below.

\begin{Def}\label{def:dec}
Let $(\Sigma,g_{ij}, K_{ij})$ be an initial data such that $g_{ij} \in L^\infty_\text{loc}\cap W^{1,2}_\text{loc}$ and $K_{ij}\in L^2_\text{loc}$.
Then we say that the data satisfies the dominant energy condition if 
\begin{equation}\label{eq:dec-muJ}
\cpl{\mu}{u}+\cpl{J}{v} \geq 0
\end{equation}
holds for any smooth, future-directed vector field $(u,v^i)$ with compact support.
\end{Def}

Using the Einstein equation, the inequality (\ref{eq:dec-muJ}) of the dominant energy condition is rewritten as 
\begin{equation}\label{eq:dec-HM}
\cpl{H_G}{u}+\cpl{M_G}{v}\geq0.
\end{equation}
Then, it is obvious that the above regularity condition for $(g_{ij}, K_{ij})$ is needed to define the energy condition.
But it is not enough to prove the positive mass theorem as Lee \& LeFloch showed \cite{LL}.
So we choose the domain for the curvature with a merely strict condition, that is,

\begin{Lem}\label{lem:ex-dc}
If $g_{ij}\in L^\infty_\text{loc}\cap W^{1,n}_\text{loc}$ and $K_{ij}\in L^n_\text{loc}$, then the distributional curvature $(H_G,M_G)$ can be defined on
\begin{equation}
X_0=\{U\in L^\frac{n}{n-2}\ |\ \underline{D}U \in L^\frac{n}{n-1}, \text{ supp}(U)\text{ : compact}\},
\end{equation} 
where $U=(u,v^i)$ and $\underline{D}U=(\underline{D}u, \underline{D}v^i)$.
\end{Lem}

\begin{proof}
By the Leibniz rule, we rewrite the integrands of the right-hand side in Eqs. (\ref{eq:def-dc}) as
\begin{subequations}
\begin{equation}\label{eq:Hi}
\frac{1}{2}\left(-V^i\underline{D}_iu\frac{dg}{dh}-V^iu\underline{D}_i\frac{dg}{dh}+(F+K_{ij}K^{ij}-K^2)u\frac{dg}{dh}\right)
\end{equation}
and
\begin{equation}
-({K^j}_i-K{\delta^j}_i)\underline{D}_jv^i\frac{dg}{dh}-({K^j}_i-K{\delta^j}_i)v^i\underline{D}_j\frac{dg}{dh}+W_{ij}^k({K^j}_k-K{\delta^j}_k)v^i\frac{dg}{dh}.
\end{equation}
\end{subequations}
Since we see $V^i\in L^n_\text{loc}$, $\underline{D}_iu\in L^{\frac{n}{n-1}}_\text{loc}$ and $\frac{dg}{dh}\in L^\infty_\text{loc}$, the power counting $\frac{1}{n}+\frac{n-1}{n}+\frac{1}{\infty}=1$ shows that the first term $V^i\underline{D}_iu\frac{dg}{dh}$ of Eq. (\ref{eq:Hi}) is integrable.
Here we used the compactness of the support $supp(u)$.
Similarly, the power counting tells us that the remaining parts are also integrable.
Here, we used the fact of $\underline{D}\frac{dg}{dh}\in L^n_\text{loc}$, $F\in L^{n/2}_\text{loc}$, $K_{ij}\in L^n_\text{loc}$ and $W^i_{jk}\in L^n_\text{loc}$, too.
\end{proof}

So we will work at the space $X_0$ defined in the above.

\section{Asymptotically flat space and generalized ADM momentum \label{sec:af-gam}}

In this section, introducing weighted functional spaces, we define $W^{1,n}_{-\tau}$-asymptotic flatness.
Then, we present the definition of ADM $(n+1)$-momentum for the current case and explore its feature.

\subsection{Definition of asymptotic flatness \label{subsec:af}}

Let $(\Sigma, h_{ij})$ be a Riemannian manifold.
We assume that there are a compact subset $C\subset \Sigma$ and an isomorphism $\Phi\colon \Sigma\setminus C\xrightarrow{\sim} \mathbb{R}^n\setminus\bm{B}(1)$ such that the induced metric $\Phi_*(h_{ij})$ coincides with the Euclidean metric $\delta_{ij}$ on $\mathbb{R}^n$, where $\bm{B}(1)$ denotes the unit ball in $\mathbb{R}^n$.
We call such a pair $(\Sigma, h_{ij})$ a background manifold.
In addition, we choose a smooth positive function $r$ on $\Sigma$ so that $r$ coincides with the ordinary radial function on $\Sigma\setminus C\approx\mathbb{R}^n\setminus \bm{B}(1)$.
Then, for $p\geq1, q\in\mathbb{R}$, we define $L^p_q=L^p_q(h)$ as the space of all measurable $u$ with finite norm
\begin{equation}
||u||_{L^p_q(h)} = \left(\int_\Sigma(|u|r^{-q})^pr^{-n}dh\right)^\frac{1}{p}.
\end{equation}
We remark that the element $u$ of these spaces could be tensor fields and $|u|$ denotes its $h$-norm.
In addition, for an integer $k\geq1$, we define $W^{k,p}_q=W^{k,p}_q(h)$ as the space of all measurable $u$ with finite norm
\begin{equation}
||u||_{W^{k,p}_q(h)} = \sum_{l\leq k}||\underline{D}^{(l)}u||_{L^p_{q-l}(h)}.
\end{equation}
For these spaces, see Ref. \onlinecite{Ba}.

Now, we are ready to define asymptotic flatness for the current case.
Following Lee \& LeFloch's study \cite{LL}, we employ the definition of asymptotic flatness which covers wide class of spacetimes.

\begin{Def}[asymptotic flatness]\label{def:af}
Let $(\Sigma, g_{ij}, K_{ij})$ be an initial data such that $g_{ij}$ is the bounded continuous field (namely $g_{ij}\in \mathcal{B}^0$).
Given $\tau>0$, we say that the data is $W^{1,n}_{-\tau}$-asymptotically flat if $g_{ij}-h_{ij}\in W^{1,n}_{-\tau}$ and $K_{ij}\in L^n_{-\tau-1}$.
\end{Def}

\subsection{Generalized ADM $(n+1)$-momentum \label{subsec:gam}}

In general, we can define the ADM $(n+1)$-momentum $(m,p_i)$ for ordinary, asymptotically flat data, that is, the pair $(m, p_i)$ is defined by
\begin{subequations}
\begin{equation}
m=\lim_{\rho\to\infty}\left(\frac{1}{2}\int_{S(\rho)}(\partial_jg_{ij}-\partial_ig_{jj})\partial_irdS\right)
\end{equation}
and
\begin{equation}
p_i=\lim_{\rho\to\infty}\int_{S(\rho)}(K_i^{\ j}-K\delta_i^{\ j})\partial_jrdS,
\end{equation}
\end{subequations}
where $S(\rho)=\{r=\rho\}$ in $\mathbb{R}^n\setminus\bm{B}(1)\approx \Sigma\setminus C$ and $dS$ the Euclidean surface measure.
For the current case, the definition of ADM $(n+1)$-momentum is given as follows.

\begin{Def}[generalized ADM $(n+1)$-momentum]\label{def:gam}
Let $(\Sigma, g_{ij}, K_{ij})$ be an asymptotically flat initial data.
Then we define the generalized ADM $(n+1)$-momentum $P$ as following:
for any smooth vector field $U=(u,v^i)$ which is constant on the frame $\Phi$, $P$ maps $U$ to
\begin{equation}
P(U)=\inf_{\varepsilon>0}\liminf_{\rho\to\infty}\left(\frac{1}{\varepsilon}\int_{\rho<r<\rho+\varepsilon} \left(\frac{1}{2}uV^i D_ir+({K^i}_j-K{\delta^i}_j)v^jD_ir\right)dh\right).
\end{equation}
\end{Def}

For time-symmetric cases ($K_{ij}=0$), this reduces to the generalized ADM mass defined by Lee \& LeFloch.

\begin{Rem}
Although $P(U)$ may not be finite, it always has a definite value.
\end{Rem}

Of course, $P(u,v^i)$ is nothing but $mu+p_iv^i$ in ordinal cases.
This is verified in next section.

\subsection{Properties of generalized ADM $(n+1)$-momentum \label{subsec:p-gam}}

In general, the generalized ADM $(n+1)$-momentum may not be split into two parts, i.e. the mass and the space-momentum.
But, under some reasonable conditions, we can do that.
We shall look at the details below.

At first, we introduce the notion ``classically $L^1$''.
Let us assume $g_{ij}\in L^\infty_\text{loc}\cap W^{1,n}_\text{loc}$ and $K_{ij}\in L^n_\text{loc}$ as in lemma \ref{lem:ex-dc}.

\begin{itemize}
\item
If there is a measurable function $\tilde{H}_G\in L^1(\Sigma, h)$ such that
\begin{equation}
\cpl{H_G}{u} = \int_\Sigma \tilde{H}_G u dh
\end{equation}
for all $u\in X_0$, then we say that $H_G$ is classically $L^1$.
The $X_0$-regularity is given in definition \ref{lem:ex-dc}.
\item
If there is a measurable covariant vector field ($\tilde{M}_G)_i\in L^1(\Sigma, h)$ such that
\begin{equation}
\cpl{M_G}{v} = \int_\Sigma (\tilde{M}_G)_iv^idh
\end{equation}
for all $v^i\in X_0$, then we say that $M_G$ is classically $L^1$.
\end{itemize}

It is clear that they are satisfied if $g_{ij}$ and $K_{ij}$ are smooth enough.

Then we have the following proposition.

\begin{Prop}
Suppose that the data $(\Sigma, g_{ij}, K_{ij})$ is $W^{1,n}_{-\tau}$-asymptotically flat with $\tau > \tau_0=\frac{n-2}{2}$.
Then, we can show the following three statements.
\begin{enumerate}
\item\label{item:splt-m}
If $H_G$ is classically $L^1$ outside some compact region, then for any $\varepsilon>0$
\begin{equation}\label{eq:def-m}
m:=P(1,0)=\lim_{\rho \to \infty}\left(\frac{1}{2\varepsilon}\int_{\{\rho<r<\rho+\varepsilon\}}V^iD_irdh\right)
\end{equation}
exists, is finite and does not depend on $\varepsilon$.
In particular, we can write the ADM $(n+1)$-momentum as
\begin{equation}
P(u,v)=mu+\inf_{\varepsilon>0}\liminf_{\rho\to\infty}\left(\frac{1}{\varepsilon}\int_{\rho<r<\rho+\varepsilon} \left(({K^i}_j-K{\delta^i}_j)v^jD_ir\right)dh\right).
\end{equation}
\item\label{item:splt-p}
If $M_G$ is classically $L^1$ outside some compact region, then for any $\varepsilon>0$
\begin{equation}\label{eq:def-sm}
p_j:=P(0,e_{(j)})=\lim_{\rho\to\infty}\left(\frac{1}{\varepsilon}\int_{\rho<r<\rho+\varepsilon} \left(({K^i}_l-K{\delta^i}_l)e_{(j)}^lD_ir\right)dh\right)
\end{equation}
exists, is finite and does not depend on $\varepsilon$, where $e_{(j)}$ is a smooth extension of $\partial_j$ in the coordinate $\Phi$ to $\Sigma$.
In particular, we can write the ADM $(n+1)$-momentum as
\begin{equation}
P(u,v)=\inf_{\varepsilon>0}\liminf_{\rho\to\infty}\left(\frac{1}{2\varepsilon}\int_{\rho<r<\rho+\varepsilon} uV^i D_irdh\right) +p_jv^j.
\end{equation}
\item\label{item:splt-mp}
If the distributional curvatures $H_G$ and $M_G$ are classically $L^1$ outside some compact region, then we can write the ADM $(n+1)$-momentum as
\begin{equation}
P(u,v)=mu+p_jv^j
\end{equation}
\end{enumerate}
\end{Prop}

Since the third statement follows directly from the first two, so we focus on the proof for the statements \ref{item:splt-m} and \ref{item:splt-p}.

Suppose that $H_G$ and $M_G$ are classically $L^1$ on a compact region $A$ and take a cut-off function
\begin{equation}\label{eq:def-chi}
\chi_\rho(x):=\chi_{\rho,\varepsilon}(x):=
\begin{cases}
1 & (r(x)\leq\rho)\\
1 - \frac{1}{\varepsilon}(r(x)-\rho) & (\rho\leq r(x)\leq \rho+\varepsilon)\\
0 & (r(x)\geq \rho+\varepsilon)
\end{cases}
\end{equation}
for $\varepsilon >0$ and a sufficiently large constant $\rho>0$ such that $A\cup C\subset\{r < \rho\}$.
Here $C$ is a compact subset of $\Sigma$ appeared in the definition of asymptotic flatness.

{\it Proof for the statement \ref{item:splt-m}:\\}
Consider $u_\rho = \chi_\rho\frac{dh}{dg}$ as a test function for $H_G$.
It is easy to see $u_\rho\in X_0$.
Then, using
\begin{equation}\label{eq:Drho}
\underline{D}_i\chi_\rho(x)=
\begin{cases}
0 & (r(x)\leq\rho \text{ or }r(x)\geq \rho+\varepsilon)\\
- \frac{1}{\varepsilon}\underline{D}_ir & (\rho\leq r(x)\leq \rho+\varepsilon),\\
\end{cases}
\end{equation}
Eq. (\ref{eq:def-dc-H}) becomes
\begin{equation}
\cpl{H_G}{u_\rho} = \frac{1}{2\varepsilon}\int_ {\rho<r<\rho+\varepsilon}V^iD_irdh+\frac{1}{2}\int_\Sigma(F+K_{ij}K^{ij}-K^2)\chi_\rho dh.
\end{equation}
On the other hand, we can simply split $\cpl{H_G}{u_\rho}$ into integrations over $A$ and $\Sigma\setminus A$ as
\begin{eqnarray}
\cpl{H_G}{u_\rho}&=&\int_A\frac{1}{2}\left(-V^i\underline{D}_i\chi_\rho+(F+K_{ij}K^{ij}-K^2)\chi_\rho\right)dh\nonumber\\
&&+\int_{\Sigma\setminus A}\frac{1}{2}\left(-V^i\underline{D}_i\left(u_\rho\frac{dg}{dh}\right)+(F+K_{ij}K^{ij}-K^2)u_\rho\frac{dg}{dh}\right)dh\nonumber\\
&=&\frac{1}{2}\int_A\left(F+K_{ij}K^{ij}-K^2\right)dh + \int_{\Sigma\setminus A}\tilde{H}_Gu_\rho dh,
\end{eqnarray}
where we used the assumption that $H_G$ is classically $L^1$ on $\Sigma\setminus A$.
Therefore we have
\begin{eqnarray}
\frac{1}{2\varepsilon}\int_ {\rho<r<\rho+\varepsilon}V^iD_irdh&=&\frac{1}{2}\int_A\left(F+K_{ij}K^{ij}-K^2\right)dh\nonumber\\
&& +\int_{\Sigma\setminus A}\tilde{H_G}u_\rho dh-\frac{1}{2}\int_\Sigma(F+K_{ij}K^{ij}-K^2)\chi_\rho dh.\label{eq:1m}
\end{eqnarray}
Since $\tilde{H}_G\in L^1(\Sigma\setminus A)$ and $F+K_{ij}K^{ij}-K^2\in L^{n/2}_{-2\tau-2}(\Sigma)\subset L^1(\Sigma)$, Eq. (\ref{eq:1m}) tells us that Eq. (\ref{eq:def-m}) has a finite limit and is independent of $\varepsilon$.
Now, we can see that the first statement holds.

{\it Proof for the statement \ref{item:splt-p}:}\\
Consider $v^i_\rho=e_{(j)}^i\chi_\rho\frac{dh}{dg}\in X_0$ as a test function for $M_G$.
Then, using Eq. (\ref{eq:Drho}), Eq. (\ref{eq:def-dc-M}) becomes
\begin{eqnarray}
\cpl{M_G}{v_\rho}&=&\frac{1}{\varepsilon}\int_{\rho<r<\rho+\varepsilon}({K^i}_l-K{\delta^i}_l)e_{(j)}^l\underline{D}_irdh\nonumber\\
&&-\int_{supp (\underline{D}e_{(j)})}({K^i}_l-K{\delta^i}_l)\chi_\rho \underline{D}_ie_{(j)}^ldh+\int_\Sigma W_{il}^k({K^l}_k-K{\delta^l}_k)e_{(j)}^i\chi_\rho dh.
\end{eqnarray}
On the other hand, we can simply split $\cpl{M_G}{v_\rho}$ into ``$A$'' and ``$\Sigma\setminus A$'' parts as
\begin{eqnarray}
\cpl{M_G}{v_\rho}&=& \int_A\left(-({K^j}_i-K{\delta^j}_i)\underline{D}_j\left(\chi_\rho e_{(j)}^i\right)+W_{ij}^k({K^j}_k-K{\delta^j}_k)\chi_\rho e_{(j)}^i\right)dh\nonumber\\
&&+ \int_{\Sigma\setminus A}\left(-({K^j}_i-K{\delta^j}_i)\underline{D}_j\left(v^i\frac{dg}{dh}\right)+W_{ij}^k({K^j}_k-K{\delta^j}_k)v^i\frac{dg}{dh}\right)dh\nonumber\\
&=&\int_A\left(-({K^i}_l-K{\delta^i}_l)\underline{D}_ie_{(j)}^l+W_{il}^k({K^l}_k-K{\delta^l}_k)e_{(j)}^i\right)dh\nonumber\\
&&+\int_{\Sigma\setminus A}(\tilde{M}_G)_iv_\rho^i dh
\end{eqnarray}
where we used the assumption that $M_G$ is classically $L^1$ on $\Sigma\setminus A$.
Therefore we have
\begin{eqnarray}
\frac{1}{\varepsilon}\int_{\rho<r<\rho+\varepsilon}({K^i}_l-K{\delta^i}_l)e_{(j)}^l\underline{D}_irdh&=&\int_A\left(-({K^i}_l-K{\delta^i}_l)\underline{D}_ie_{(j)}^l+W_{il}^k({K^l}_k-K{\delta^l}_k)e_{(j)}^i\right)dh\nonumber\\
&&+\int_{\Sigma_\setminus A}(\tilde{M}_G)_iv^i_\rho dh+\int_{supp (\underline{D}e_{(j)})}({K^i}_l-K{\delta^i}_l)\chi_\rho \underline{D}_i{e_{(j)}}^ldh\nonumber\\
&&-\int_\Sigma W_{il}^k({K^l}_k-K{\delta^l}_k)e_{(j)}^i\chi_\rho dh.\label{eq:2m}
\end{eqnarray}
Since $(\tilde{M}_G)_i\in L^1(\Sigma\setminus A)$, $W_{il}^k({K^l}_k-K{\delta^l}_k)\in L^{n/2}_{-2\tau-2}(\Sigma)\subset L^1(\Sigma)$ and $supp (\underline{D}e_{(j)})$ is compact, Eq. (\ref{eq:2m}) tells us that Eq. (\ref{eq:def-sm}) has a finite limit and is independent of $\varepsilon$.
So the statement  \ref{item:splt-p} is also proven.

This proposition is restricted to data with classically $L^1$ curvature.
For general data, we have

\begin{Lem}\label{lem:gam-lim}
Let $(\Sigma, g_{ij}, K_{ij})$ be a $W^{1,n}_{-\tau}$-asymptotically flat data with $\tau>\tau_0=\frac{n-2}{2}$.
If it satisfies the dominant energy condition (see definition \ref{def:dec}), then, for any $\varepsilon>0$ and any future-directed vector $(u,v)$ which is constant on the frame $\Phi$,
\begin{equation}
P(u,v)=\lim_{\rho\to\infty}\left(\frac{1}{\varepsilon}\int_{\rho<r<\rho+\varepsilon} \left(\frac{1}{2}uV^i D_ir+({K^i}_j-K{\delta^i}_j)v^jD_ir\right)dh\right)
\end{equation}
exists and does not depend on $\varepsilon>0$.
In particular, if $P(u,v)=0$ for any future-directed $(u,v)$, then $P(u,v)=0$ for all $(u,v)$.
\end{Lem}

\begin{proof}
Consider $(u_\rho, v^i_\rho)=(u\chi_\rho\frac{dh}{dg},v^i\chi_\rho\frac{dh}{dg})$ as a test function for $(H_G,M_G)$.

(i) At first, we show that $\cpl{H_G}{u_\rho}+\cpl{M_G}{v_\rho}$ is monotonically increasing and has the limit
\begin{equation}\label{eq:HM}
\lim_{\rho\to\infty}\left(\cpl{H_G}{u_\rho}+\cpl{M_G}{v_\rho}\right).
\end{equation}
Let $\rho_1<\rho_2$.
Then it is easily checked that $(u_{\rho_2}-u_{\rho_1}, v_{\rho_2}-v_{\rho_1})$ is also future-directed.
So the dominant energy condition implies
\begin{eqnarray}
\left(\cpl{H_G}{u_{\rho_2}}+\cpl{M_G}{v_{\rho_2}}\right)-\left(\cpl{H_G}{u_{\rho_1}}+\cpl{M_G}{v_{\rho_1}}\right)\nonumber\\
=\cpl{H_G}{u_{\rho_2}-u_{\rho_1}}+\cpl{M_G}{v_{\rho_2}-v_{\rho_1}}\geq0,
\end{eqnarray}
that is, $\cpl{H_G}{u_\rho}+\cpl{M_G}{v_\rho}$ is monotonically increasing.

(ii) Next, we show that the limit (\ref{eq:HM}) is independent of $\varepsilon>0$.
To do so, we write $\varepsilon$ explicitly as $(u_{\rho,\varepsilon},v_{\rho,\varepsilon}^i)=(u\chi_{\rho,\varepsilon}\frac{dh}{dg},v^i\chi_{\rho,\varepsilon}\frac{dh}{dg})$.
We take arbitrary $\varepsilon_1,\varepsilon_2>0$ and $\rho_1$, and choose $\rho_2$ such that $\rho_2>\rho_1+\varepsilon_1$.
Then $(u_{\rho_2,\varepsilon_2}-u_{\rho_1,\varepsilon_1}, v_{\rho_2,\varepsilon_2}^i-v_{\rho_1,\varepsilon_1}^i)$ is also future-directed.
So, by the dominant energy condition, we have 
\begin{eqnarray}
\left(\cpl{H_G}{u_{\rho_2,\varepsilon_2}}+\cpl{M_G}{v^i_{\rho_2,\varepsilon_2}}\right)-\left(\cpl{H_G}{u_{\rho_1,\varepsilon_1}}+\cpl{M_G}{v^i_{\rho_1,\varepsilon_1}}\right)\nonumber\\
=\cpl{H_G}{u_{\rho_2,\varepsilon_2}-u_{\rho_1,\varepsilon_1}}+\cpl{M_G}{v_{\rho_2,\varepsilon_2}^i-v_{\rho_1,\varepsilon_1}^i}\geq0.
\end{eqnarray}
This implies
\begin{equation}
\lim_{\rho_2\to\infty}\left(\cpl{H_G}{u_{\rho_2,\varepsilon_2}}+\cpl{M_G}{v^i_{\rho_2,\varepsilon_2}}\right)\geq\cpl{H_G}{u_{\rho_1,\varepsilon_1}}+\cpl{M_G}{v^i_{\rho_1,\varepsilon_1}}.
\end{equation}
In particular,
\begin{equation}
\lim_{\rho_2\to\infty}\left(\cpl{H_G}{u_{\rho_2,\varepsilon_2}}+\cpl{M_G}{v^i_{\rho_2,\varepsilon_2}}\right)\geq\lim_{\rho_1\to\infty}\left(\cpl{H_G}{u_{\rho_1,\varepsilon_1}}+\cpl{M_G}{v^i_{\rho_1,\varepsilon_1}}\right).
\end{equation}
Since $\varepsilon_1$ and $\varepsilon_2$ are arbitrary, this inequality implies that the limit (\ref{eq:HM}) is independent of $\varepsilon>0$.

(iii) By using (i) and (ii), we complete our proof of this lemma.
From Eqs. (\ref{eq:1m}) and (\ref{eq:2m}), we have
\begin{eqnarray}
&&\frac{1}{\varepsilon}\int_ {\rho<r<\rho+\varepsilon}\left(\frac{1}{2} uV^iD_ir+({K^i}_j-K{\delta^i}_l)v^j\underline{D}_ir\right)dh\nonumber\\
&=&\cpl{H_G}{u_\rho}+\cpl{M_G}{v_\rho}-\frac{1}{2}\int_\Sigma(F+K_{ij}K^{ij}-K^2)u_\rho dh\nonumber\\
&&\ \ +\int_{supp (\underline{D}v)}({K^i}_j-K{\delta^i}_j)\underline{D}_iv_\rho^jdh-\int_\Sigma W_{il}^k({K^l}_k-K{\delta^l}_k)v^i_\rho dh.\label{eq:Prho}
\end{eqnarray}
Here, note that $F+K_{ij}K^{ij}-K^2$ and $W_{il}^k({K^l}_k-K{\delta^l}_k)$ belong to $L^{n/2}_{-2\tau-2}(\Sigma)\subset L^1(\Sigma)$ and $supp(\underline{D}v_\rho)$ is compact.
So, together with (i) and (ii), it is directly shown that the right-hand side of Eq. (\ref{eq:Prho}) has a limit independent of $\varepsilon>0$.

The last part is obvious.
\end{proof}

\section{Lichnerowicz-Weitzenb{\"o}ck formula \label{sec:lwf}}

Since we will prove the main theorem using spinor, we introduce spinor bundle and spin connections, and then we show Lichnerowicz-Weitzenb{\"o}ck formula for distributional curvature, which will be a key part of the proof for our main theorem.

\subsection{Spin connections \label{subsec:lwf-sc}}

At first, we introduce spinor bundle we work on.

Suppose that $(\Sigma, h_{ij})$ has a spin structure, that is, there exist a principal $Spin_n$ bundle for the cotangent bundle $T^*M$ with the metric $h^{ij}$.
As Lee \& LeFloch showed \cite{LL}, one can regard the $Spin_n$ structure as that of $g_{ij}\in C^0$.
Then, one extends this $Spin_n$ structure to a $Spin_{n,1}$ structure \cite{D} and constructs spinor bundle $S$ using this $Spin_{n,1}$ structure.

For convenience, we fix a local frame of the principal $Spin_{n,1}$ bundle and consider the corresponding local frames $\underline{e}^i$, $e^i$ and $\psi_I$ for $(T^*\Sigma, h^{ij})$, $(T^*\Sigma,g^{ij})$ and $S$, respectively.
The subscript ``$I$'' of $\psi_I$ denotes the label of spinor and takes values $1,\cdots 2^{[\frac{n+1}{2}]}$.
In the below, we write all tensors by index notation with respect to the frame $e^i$ and regard a spinor as a column vector.
For example, the inner product for spinors $\psi,\phi$ is expressed as $(\psi,\phi)=\psi^\dagger\phi$, where dagger stands for the hermitian conjugate.

On the bundle $S$, we have three spin connections $\underline{D}$, $D$ and $\nabla$.
The connections are defined by
\begin{equation}
\underline{D}\psi=\partial\psi-\frac{1}{4}\underline{\bm{\omega}}_{ij}\underline{c}(\underline{e}^i)\underline{c}(\underline{e}^j)\psi
\end{equation}
using the connection $1$-form $\underline{\bm{\omega}}_i^{\ j}$ of $h_{ij}$,
\begin{equation}
D\psi=\partial\psi-\frac{1}{4}\bm{\omega}_{ij}c(e^i)c(e^j)\psi
\end{equation}
using the connection $1$-form $\bm{\omega}_i^{\ j}$ of $g_{ij}$ and
\begin{equation}
\nabla_i\psi=D_i\psi+\frac{1}{2}K_{ij}c(e^0)c(e^j)\psi
\end{equation}
using the data $K_{ij}$, where $\underline{c}(\cdot)$ and $c(\cdot)$ denote the Clifford actions on $S$ as spinor bundle for $h^{ij}$ and $g^{ij}$, respectively.
In addition, the action $c(e^0)$ corresponds to the Clifford action of the future-directed unit normal of $\Sigma$ in spacetime.

Then,  we have

\begin{Lem}\label{lem:spn-n}
Let $(\Sigma,g_{ij},K_{ij})$ be a $W^{1,n}_{-\tau}$-asymptotically flat data for $\tau>\tau_0=\frac{n-2}{2}$ with a spin structure.
Then the operator $\nabla:W^{1,2}_{-\tau_0}\to L^2_{-\tau_0-1}$ is bounded.
\end{Lem}
\begin{proof}
Let $A_i=\nabla_i-\underline{D}_i=\frac{1}{4}\underline{\bm{\omega}}_{ij}\underline{c}(\underline{e}^i)\underline{c}(\underline{e}^j)-\frac{1}{4}\bm{\omega}_{ij}c(e^i)c(e^j)+\frac{1}{2}K_{ij}c(e^0)c(e^j)$.
Since we have
\begin{equation}
||\nabla\psi||_{L^2_{-\tau_0-1}}=||\underline{D}\psi+A\psi||_{L^2_{-\tau_0-1}}\leq||\underline{D}\psi||_{L^2_{-\tau_0-1}}+||A\psi||_{L^2_{-\tau_0-1}},
\end{equation}
it is sufficient to estimate the last term $||A\psi||_{L^2_{-\tau_0-1}}$ for the proof.
From the asymptotic flatness, it can be proven that $A$ belongs to $L^n_{-\tau-1}\subset L^n_{-\tau_0-1}$.
So we compute
\begin{eqnarray}
||A\psi||_{L^2_{-\tau_0-1}}&\leq&||A||_{L^n_{-\tau_0-1}}||\psi||_{L^\frac{2n}{n-2}_{-\tau_0}}\\
&\leq&C||\, |\psi|\, ||_{W^{1,2}_{-\tau_0}}\\
&\leq&C||\psi||_{W^{1,2}_{-\tau_0}}
\end{eqnarray}
for $C=||A||_{L^n_{-\tau_0-1}}$, where we used the weighted H{\"o}lder inequality, the weighted Sobolev inequality \cite{Ba} and Kato's inequality for the each lines.
\end{proof}

\subsection{Lichnerowicz-Weitzenb{\"o}ck formula \label{subsec:lwf-lwf}}

Now it is ready to show the Lichnerowicz-Weitzen{\"o}ck formula for distributional curvature.

At first, from the pedagogical point of view, we suppose that the data $(\Sigma,g_{ij},K_{ij})$ is smooth enough.
In this case, we have the ordinary Lichnerowicz-Weitzenb{\"o}ck formula for $\nabla$ as \cite{W}
\begin{equation}\label{eq:lwf}
\snabla^2\phi = -\delta^{ij}\nabla_i\nabla_j\phi +\frac{1}{2}(H_G+(M_G)_ic(e^0)c(e^i))\phi,
\end{equation}
where $\snabla:=c(e^i)\nabla_i$ is the Dirac operator, $\phi$ is a smooth spinor field and $(H_G,M_G)$ is the curvature defined classically (see Eqs. (\ref{eq:def-HM})).
Although this is not current case, it is nice to see more.
This is because such consideration gives us a hint for the current distributional cases.

Multiplying another spinor field $\psi$ with Eq. (\ref{eq:lwf}), we get
\begin{equation}
(\psi,\snabla^2\phi)=-(\psi,\delta^{ij}\nabla_i\nabla_j\phi)+\frac{1}{2}\left(H_G(\psi,c(e^0)\phi)+(M_G)_i(\psi,c(e^0)c(e^i)\phi)\right).
\end{equation}
Then, we suppose that $\phi$ has a compact support and integrate this equation on the whole space $\Sigma$.
Then we establish the integrated version of the Lichnerowicz-Weitzenb\"ock formula
\begin{equation}\label{lwf-i}
(\snabla\psi,\snabla\phi)_{L^2}=(\nabla\psi,\nabla\phi)_{L^2}+\frac{1}{2}\int_\Sigma \left(H_G u+(M_G)_iv^i\right)dg,
\end{equation}
where we set $u=(\psi,\phi)$ and $v^i=(\psi,c(e^0)c(e^i)\phi)$.
In the above, we used the integration by part.

From the above observation for classical case, we shall arrive at

\begin{Lem}[Lichnerowicz-Weitzenb\"ock formula for distributional curvature]\label{lem:lwf-dc}
Assume $g_{ij}\in C^0\cap W^{1,n}_\text{loc}$, $K_{ij}\in L^n_\text{loc}$ and $\psi,\phi\in W^{1,2}_\text{loc}$.
If $\phi$ has a compact support, then we have
\begin{equation}\label{eq:lwf-dc}
0=-(\snabla\psi,\snabla\phi)_{L^2}+(\nabla\psi, \nabla\phi)_{L^2}+\frac{1}{2}\left(\cpl{H_G}{(\psi,\phi)}+\cpl{M_G}{(\psi,c(e^0)c(e^i)\phi)}\right).
\end{equation}
\end{Lem}

\begin{proof}
By density argument, it is easy to show the formula from the integrated version (\ref{lwf-i}).
Here we implicitly used the Sobolev embedding theorem $W^{1,2}_\text{loc}\subset L^{\frac{2n}{n-2}}_\text{loc}$ for $((\psi,\phi),(\psi,c(e^0)c(e^i)\phi))\in X_0$ and lemma \ref{lem:spn-n} for the treatments of $\nabla$.
\end{proof}

Following Witten \cite{W}, we derive the formula for the cases with asymptotic boundary terms.
Suppose that the manifold $\Sigma$ has a  background data $(h_{ij}, \Phi)$ and let $L^i$ be the operator defined by
\begin{equation}
L^i:=(c(e^i)c(e^j)+\delta^{ij})\nabla_j.
\end{equation}
Then we have

\begin{Lem}\label{lem:lwf-dc-b}
Assume $g_{ij}\in C^0\cap W^{1,n}_\text{loc}$ and $K_{ij}\in L^n_\text{loc}$.
Then, for any spinor field $\psi$, $\varepsilon>0$ and a sufficiently large $\rho>0$,
\begin{eqnarray}
\frac{1}{\varepsilon}\int_{\rho<r<\rho+\varepsilon}(L^i\psi,\psi)D_irdg&=&-(\snabla\psi,\chi_\rho\snabla\psi)_{L^2}+(\nabla\psi,\chi_\rho\nabla\psi)_{L^2}\nonumber\\
&&+\frac{1}{2}\left(\cpl{H_G}{\chi_\rho(\psi,\psi)}+\cpl{M_G}{(\psi,c(e^0)c(e^i)\psi}\right)
\label{eq:lwf-dc-b}
\end{eqnarray}
holds.
\end{Lem}
\begin{proof}
By lemma \ref{lem:lwf-dc} with $\phi=\chi_\rho\psi$, where $\chi_\rho$ is defined by Eq. (\ref{eq:def-chi}), we have
\begin{eqnarray}
0&=&-(\snabla\psi,\snabla(\chi_\rho\psi))_{L^2}+(\nabla\psi, \nabla(\chi_\rho\psi))_{L^2}\nonumber\\
&&+\frac{1}{2}\left(\cpl{H_G}{\chi_\rho(\psi,\psi)}+\cpl{M_G}{\chi_\rho(\psi,c(e^0)c(e^i)\psi)}\right).\label{eq:lwf-rho}
\end{eqnarray}
Since $\snabla(\chi_\rho\psi)=\chi_\rho\snabla\psi+D_i\chi_\rho c(e^i)\psi$, we can rewrite the first term of the right-hand side as
\begin{equation}
(\snabla\psi,\snabla(\chi_\rho\psi))_{L^2}=(\snabla\psi,\chi_\rho\snabla\psi)_{L^2}+\frac{1}{\varepsilon}\int_{\rho<r<\rho+\varepsilon}(c(e^i)c(e^j)\nabla_j\psi,\psi)D_irdg.
\end{equation}
In the above, we used the triviality of $\underline{D}_i\chi_\rho$ except for $\{\rho<r<\rho+\varepsilon\}$.
In a similar way, the second term becomes
\begin{equation}
(\nabla\psi, \nabla(\chi_\rho\psi))_{L^2}=(\nabla\psi,\chi_\rho\nabla\psi)_{L^2}-\frac{1}{\varepsilon}\int_{\rho<r<\rho+\varepsilon}\delta^{ij}(\nabla_i\psi,\psi)D_jrdg.
\end{equation}
Then, a minor rearrangement of Eq. (\ref{eq:lwf-rho}) gives us Eq. (\ref{eq:lwf-dc-b}).
\end{proof}

Next, we pick up the informations of the ADM $(n+1)$-momentum from the left-hand side of Eq. (\ref{eq:lwf-dc-b}).

\begin{Lem}\label{lem:l-gam}
Let $(\Sigma, g_{ij}, K_{ij})$ be a $W^{1,n}_{-\tau}$-asymptotically flat data with $\tau>\tau_0$ and take a spinor field $\psi_0$ which is constant on the frame $\Phi$.
Then, for any spinor field $\psi$ with $\epsilon:=\psi-\psi_0\in W^{1,2}_{-\tau_0}$, we have
\begin{eqnarray}
&&\frac{1}{\varepsilon}\int_{\rho<r<\rho+\varepsilon}(L^i\psi,\psi)D_irdg\nonumber\\
\label{eq:l-gam}
&=&\int_{\rho<r<\rho+\varepsilon}\left(\frac{\psi_0,c(e^0)\psi_0}{4\varepsilon}V^iD_ir+\frac{\psi_0,c(e^j)\psi_0}{2\varepsilon}({K^i}_j-K{\delta^i}_j)D_ir\right)dh+\int_{\rho<r<\rho+\varepsilon}udh
\end{eqnarray}
for some function $u\in L^1_{-2\tau_0-1}$.
\end{Lem}
\begin{proof}
Take the frame $\underline{e}^i=dx^i$ of $\Phi$ and consider the corresponding frames $\psi_I$ and $e^i$ for spinor bundle $S$ and $(T^*\Sigma, g^{ij})$ respectively.
The existence of $\psi_I$ is guaranteed by the topology of $\mathbb{R}^n\setminus\bm{B}(1)$.
This is explained as following:
since the $Spin_n$ structure is trivial on the $\Sigma\setminus C\approx\mathbb{R}^n\setminus\bm{B}(1)$, the $Spin_{n,1}$ structure is also trivial on that region.
This implies that one can consider such frame $\psi_I$.
Next, let $\underline{e}_i=\partial_i$ and $e_i$ be the duals of $\underline{e}^i$ and $e^i$ respectively.
Since it is directly proven that $e^i-dx^i\in W^{1,n}_{-\tau}$ and $e_i-\partial_i\in W^{1,n}_{-\tau}$ hold, we can express the connection $1$-form $\bm{\omega}_{ij}$ as $\bm{\omega}_{ij}(e_k)=\frac{1}{2}(\partial_ig_{jk}-\partial_jg_{ik})+O(L^{n/2}_{-2\tau-1})$.

Now, using $\epsilon=\psi-\psi_0$, we rewrite the integrand of the left-hand side in Eq. (\ref{eq:l-gam}) as
\begin{equation}\label{eq:LppDr}
(L^i\psi,\psi)D_ir=(L^i\psi_0,\psi_0)D_ir-D_j\left(\epsilon,(c(e^i)c(e^j)+\delta^{ij})\psi_0\right)D_ir+O(L^1_{-2\tau_0-1}).
\end{equation}
We remark that the integration of the second term in the above is zero.
This is because
\begin{eqnarray}
\int_{\rho<r<\rho+\varepsilon}D_j\left(\epsilon,(c(e^i)c(e^j)+\delta^{ij})\psi_0\right)D_irdg&=&\int_{+\{r=\rho+\varepsilon\}-\{r=\rho\}}(\epsilon,(c(e^i)c(e^j)+\delta^{ij})\psi_0)D_irD_jrd\sigma\nonumber\\
&=&0,
\end{eqnarray}
where $d\sigma$ denotes the volume measure for the induced metric from $g_{ij}$ and we used the anti-symmetry of $c(e^i)c(e^j)+\delta^{ij}$ with respect to the indices in the last equality.
Next, using the expression of the connection $1$-form above, we rewrite the first term of the left-hand side in Eq. (\ref{eq:LppDr}) as
\begin{equation}\label{eq:inte}
(L^i\psi_0,\psi_0)D_ir=\frac{(\psi_0,\psi_0)}{4}V^iD_ir+\frac{(\psi_0,c(e^0)c(e^j)\psi_0)}{2}(K^i_{\ j}-K\delta^i_{\ j})D_ir+O(L^1_{-2\tau_0-1}).
\end{equation}
Since the asymptotic flatness of the data implies that the contribution from the difference of the measures $dg$ and $dh$ can be absorbed into $O(L^1_{-2\tau_0-1})$, we can see that Eq. (\ref{eq:l-gam}) holds.
\end{proof}

\section{The proof of the main theorem \label{sec:mt}}

To prove our main theorem, we consider the Dirac-Witten equation for a spinor field $\psi$,
\begin{equation}\label{eq:dw}
\snabla\psi=0.
\end{equation}
Then, the next theorem guarantees the existence of the solutions to the Dirac-Witten equation.

\begin{Thm}\label{thm:do-iso}
Let $(\Sigma,g_{ij},K_{ij})$ be a $W^{1,n}_{-\tau_0}$-asymptotically flat data with $\tau>\tau_0=\frac{n-2}{2}$.
If this data satisfies the dominant energy condition, then the operator
\begin{equation}
\snabla:W^{1,2}_{-\tau_0}\to L^2_{-\tau_0-1}=L^2
\end{equation}
is an isomorphism.
\end{Thm}
\begin{proof}[Proof of Theorem \ref{thm:do-iso}]
Since $c(e^i)$ $(i=1,\cdots, n)$ act on spinors as an unitary with respect to the product $(\cdot,\cdot)$ for spinors, it is obvious that $\snabla$ defines a bounded linear map from $W^{1,n}_{-\tau_0}$ to $L^2_{-\tau_0-1}$, that is, there exists a constant $c$ such that
\begin{equation}\label{eq:do-bdd}
||\snabla\psi||_{L^2_{-\tau_0-1}}\leq c||\psi||_{W^{1,2}_{-\tau_0}}
\end{equation}
for any spinor field $\psi\in W^{1,2}_{-\tau_0}$.
Next, we show that there exists a constant $C$ such that
\begin{equation}\label{eq:snp}
||\psi||_{W^{1,2}_{-\tau_0}}\leq C||\snabla\psi||_{L^2_{-\tau_0-1}}
\end{equation}
for any spinor field $\psi\in W^{1,2}_{-\tau_0}$.
Now, we have Eq. (\ref{eq:lwf-dc-b}).
The integrand of the left-hand side of the equation belongs to $L^1_{-2\tau_0-1}$.
So, by lemma 4.1 in Ref. \onlinecite{LL}, one can take a sequence $\rho_k$ such that $\rho_k\to\infty$ and $\int_{\rho_k<r<\rho_k+\varepsilon} (L^i\psi,\psi)D_irdg\to0$ as $k\to\infty$.
In addition, since the vector $((\psi,\psi), (\psi,c(e^0)c(e^i)\psi))$ is future-directed, we have
\begin{eqnarray}
&&-(\snabla\psi,\chi_{\rho_k}\snabla\psi)_{L^2}+(\nabla\psi,\chi_{\rho_k}\nabla\psi)_{L^2}\nonumber\\
&\leq&-(\snabla\psi,\chi_{\rho_k}\snabla\psi)_{L^2}+(\nabla\psi,\chi_{\rho_k}\nabla\psi)_{L^2}+\frac{1}{2}\left(\cpl{H_G}{\chi_{\rho_k}(\psi,\psi)}+\cpl{M_G}{\chi_{\rho_k}(\psi,c(e^0)c(e^i)\psi}\right)\nonumber\\
&=&\frac{1}{\varepsilon}\int_{\rho_k<r<\rho_k+\varepsilon}(L^i\psi,\psi)D_irdg\xrightarrow{k\to\infty}0,
\end{eqnarray}
where we used the dominant energy condition at the second line and Eq. (\ref{eq:lwf-dc-b}) at the third equality.
From this, we see
\begin{equation}\label{eq:snp2}
||\nabla\psi||_{L^2}\leq||\snabla\psi||_{L^2}.
\end{equation}
In addition, we have the weighted Poincar\'e inequality (see theorem 9.5 in Ref. \onlinecite{BC}, or proposition 2.6 and Eq. (2.10) in Ref. \onlinecite{An})
\begin{equation}\label{eq:Poin}
||\psi||_{W^{1,2}_{-\tau_0}}\leq C||\nabla\psi||_{L^2_{-\tau_0-1}}
\end{equation}
with some constant $C>0$.
Then, by combining the inequalities \eqref{eq:snp2} and \eqref{eq:Poin}, we obtain Eq. (\ref{eq:snp}).

The remaining task is to show the surjectivity of $\snabla$.
This will be done in a way similar to the proof for proposition 4.2 in Ref. \onlinecite{LL}.

At first, we consider $H=W^{1,2}_{-\tau_0}$ with an inner product $(\psi,\phi)_H=(\snabla\psi,\snabla\phi)_{L^2}$.
Since we have Eqs. (\ref{eq:do-bdd}) and (\ref{eq:snp}), $H$ is a Hilbert space.
Now, take a continuous spinor field $\eta$ with compact support.
By the Riesz representation theorem for $H$ (for example, see Ref. \onlinecite{K}), there exists $\omega\in W^{1,2}_{-\tau_0}$ such that $(\eta,\phi)_{L^2}=(\snabla\omega,\snabla\phi)_{L^2}$.
Then, we set $\xi = \snabla\omega\in L^2$.
We will show $\xi\in W^{1,2}_{-\tau_0}$ and $\eta=\snabla\xi$.
We choose a sequence $\xi_k\in W^{1,2}_{-\tau_0}$ converging to $\xi$ in $L^2$.
It is easily checked that $\snabla\xi_k$ weakly converges to $\eta$ in $L^2$.
In particular, $||\snabla\xi_k||_{L^2}$ is bounded.
This implies that a sequence $||\xi_k||_{W^{1,2}_{-\tau_0}}$ is also bounded.
So, there exists a subsequence $\xi_{k_l}$ weakly converging in $W^{1,2}_{-\tau_0}$.
The weak limit of this sequence must be $\xi$.
This means $\xi\in W^{1,2}_{-\tau_0}$.
In addition, we have
\begin{equation}
(\snabla\xi,\phi)_{L^2}=(\xi,\snabla\phi)_{L^2}=(\snabla\omega,\snabla\phi)_{L^2}=(\eta,\phi)_{L^2}
\end{equation}
for any spinor field $\phi\in W^{1,2}_{-\tau_0}$ with compact support.
This tells us $\snabla\xi=\eta$.
Therefore, by density argument, it is proved that $\snabla:W^{1,2}_{-\tau_0}\to L^2_{-\tau_0-1}$ is surjective.
\end{proof}

By using this theorem, one can show the existence of the solutions to Eq. (\ref{eq:dw}).

\begin{Cor}\label{cor:slv-dw}
Assume that our data $(\Sigma,g_{ij},K_{ij})$ is $W^{1,2}_{-\tau}$-asymptotically flat for $\tau>\tau_0=\frac{n-2}{2}$ and satisfies the dominant energy condition.
Let $\psi_0$ be a smooth spinor field which is constant on the frame $\Phi$.
Then, there exists a solution $\psi\in W^{1,2}_\text{loc}$ of Eq. (\ref{eq:dw}) such that $\psi-\psi_0\in W^{1,2}_{-\tau_0}$.
\end{Cor}
\begin{proof}
It is easily proven that $-\snabla\psi_0\in L^2_{-\tau_0-1}$ holds.
So, by theorem \ref{thm:do-iso}, $\epsilon=-\snabla^{-1}\snabla\psi_0$ exists uniquely.
Then we have $\psi=\psi_0+\epsilon$ as a solution of Eq. (\ref{eq:dw}).
\end{proof}

Finally, we state the proof of the main theorem.

\begin{proof}[Proof of Theorem \ref{thm:main}]
Take a smooth spinor filed $\psi_0$ which is constant on the frame $\Phi$.
By applying corollary \ref{cor:slv-dw} to this $\psi_0$, we construct a spinor field $\psi$.
Then, by lemma \ref{lem:lwf-dc-b} and lemma \ref{lem:l-gam}, we have 
\begin{eqnarray}
&&\int_{\rho<r<\rho+\varepsilon}\left(\frac{(\psi_0,\psi)}{4\varepsilon}V^iD_ir+\frac{(\psi,c(e^0)c(e^j)\psi)}{2\varepsilon}({K^i}_j-K{\delta^i}_j)D_ir\right)dh\nonumber\\
&&=(\nabla\psi,\chi_\rho\nabla\psi)_{L^2}+\frac{1}{2}\left(\cpl{H_G}{\chi_\rho(\psi,\psi)}+\cpl{M_G}{(\psi,c(e^0)c(e^i)\psi}\right)+\int_{\rho<r<\rho+\varepsilon}udh\label{eq:pmt1}
\end{eqnarray}
for some function $u\in L^1_{-2\tau_0-1}$.
Let $\rho_k$ be the sequence guaranteed by lemma 4.1 of Ref. \onlinecite{LL} for this function $u$, that is, a sequence such that $\rho_k\to\infty$ and $\int_{\rho_k<r<\rho_k+\varepsilon}udg\to0$ as $k\to\infty$.
By replacing $\rho$ by $\rho_k$ in Eq. (\ref{eq:pmt1}) and taking a limit $k\to\infty$, the last term in that equation vanishes and we obtain
\begin{align}\label{eq:pmt3}
\lim_{k\to\infty}\int_{\rho_k<r<\rho_k+\varepsilon}&\left(\frac{(\psi_0,\psi_0)}{4\varepsilon}V^iD_ir+\frac{(\psi_0,c(e^0)c(e^j)\psi_0)}{2\varepsilon}({K^i}_j-K{\delta^i}_j)D_ir\right)dh\nonumber\\
&=(\nabla\psi,\nabla\psi)_{L^2}+\frac{1}{2}\lim_{k\to\infty}\frac{1}{2}\left(\cpl{H_G}{\chi_{\rho_k}(\psi,\psi)}+\cpl{M_G}{\chi_{\rho_k}(\psi,c(e^0)c(e^i)\psi}\right).
\end{align}
By lemma \ref{lem:gam-lim}, it is proven that the left-hand side is nothing, but $P((\psi_0,\psi_0),(\psi_0c(e^0)c(e^i)\psi_0))$.
In addition, by the dominant energy condition, this is non-negative.
Since the choice of $\psi_0$ is arbitrary, this implies that $P(U)$ is non-negative for any future-directed vector field $U$ which is constant on the frame $\Phi$.

At the last, we suppose $P$ is zero, that is, $P(U)=0$ for any above $U$.
Then, Eq. (\ref{eq:pmt3}) implies $\nabla\psi=0$.
This means that $\psi$ is parallel with respect to the connection $\nabla$.
\end{proof}

\section{An example -spaces with corner- \label{sec:ex}}

In this section, we consider a data with spacetime corner such that it satisfies the dominant energy condition in our distributional sense.

We consider a bonded data $(\Sigma,g_{ij},K_{ij})$ from two smooth (at least $C^2$) data $(\Sigma_\pm, (g_\pm)_{ij}, (K_\pm)_{ij})$ with isometry $\partial\Sigma_+\approx\partial\Sigma_-$, that is, $(\Sigma,g_{ij})$ is a Riemannian manifold $\Sigma_+\cup\Sigma_-$ with identification $S:=\partial\Sigma_+=\partial\Sigma_-$ and $K_{ij}$ a tensor field on $\Sigma$ such that $K_{ij}=(K_\pm)_{ij}$ on $\Sigma_\pm$ respectively.
We remark that $K_{ij}$ is multivalued on the surface $S$, that is, $K_{ij}$ has values $(K_+)_{ij}$ and $(K_-)_{ij}$ on $S$.
In general, for such multivalued quantity $A$ on $S$, we define $A_\Delta$ by $A_\Delta:=A_+-A_-$.
Next, we take the unit normal vector field $n^i$ of the surface $S$ pointing toward the interior of $\Sigma_+$ and let $k$ be its mean curvature.
Then $k$ is also multivalued on $S$.
One is $k_+$ for $S\subset\Sigma_+$ and the other is $k_-$ for $S\subset\Sigma_-$.
In this notation, we consider a covariant vector $(p_0, p_i)$ defined by the pair of the scalar $p_0=-k_\Delta$ and the covariant vector $p_i=((K_\Delta)_{ij}-K_\Delta\delta_{ij})n^j$.
This is nothing, but the Hamiltonian momentum density for the surface.
So, it is reasonable to define the causalness of surfaces as below.

\begin{Def}
We say that the surface $S$ is a causal corner if the vector $(p_0,p_i)$ is future-directed i.e. $p_0\geq\sqrt{g^{ij}p_ip_j}$.
\end{Def}

When $K_{ij}=0$, this is reduced to Riemannian case (see Ref. \onlinecite{M}).

\begin{Prop}
Assume that the data $(\Sigma_\pm,(g_\pm)_{ij},(K_\pm)_{ij})$ is smooth (at least $C^2$) and satisfies the dominant energy condition in the classical sense.
If the bonding surface $S$ is a causal corner, then the bonded data $(\Sigma,g_{ij}, K_{ij})$ satisfies the dominant energy condition in the distributional sense.
\end{Prop}

\begin{proof}
For convenience, we extend the normal $n^i$ so that $D_nn^i=0$ (geodesic) and let $z$ be its affine parameter such that $z=0$ on $S$.
Then, we extend the function $z$ to a smooth function on $\Sigma$ so that $\text{int}(\Sigma^+)=\{z>0\}$ and $\text{int}(\Sigma_-)=\{z<0\}$.
Next, take an auxiliary metric $h_{ij}$ so that $n^i$ is also the unit normal of $S$ with respect to this $h_{ij}$ and geodesic (i.e. $\underline{D}_nn^i=0$) on some neighborhood of $S$.
In addition, we can suppose that the metrics of $S$ induced from $h_{ij}$ and $g_{ij}$ coincide.
Under this setting, we take a smooth future-directed vector field $(u,v)$.
Then $\cpl{H_G}{u}+\cpl{M_G}{v}$ is expressed as
\begin{eqnarray}
\cpl{H_G}{u}+\cpl{M_G}{v}=\lim_{\varepsilon\to0}\left\{\int_{|z|>\varepsilon}\frac{1}{2}\left(-V^i\underline{D}_i(u\frac{dg}{dh})+(F+K_{ij}K^{ij}-K^2)u\frac{dg}{dh}\right)dh\right.\nonumber\\
\left.+\int_{|z|>\varepsilon}\left(-({K^j}_i-K{\delta^j}_i)\underline{D}_j(v^i\frac{dg}{dh})+W_{ij}^k({K^j}_k-K{\delta^j}_k)v^i\frac{dg}{dh}\right)dh\right\}.
\end{eqnarray}
By using integration by parts, we rewrite this as
\begin{eqnarray}
\cpl{H_G}{u}+\cpl{M_G}{v}&=&\lim_{\varepsilon\to0}\left\{\int_{|z|>\varepsilon}\frac{1}{2}(\underline{D}_iV^i+F+K_{ij}K^{ij}-K^2)udg\right.\nonumber\\
&&+\int_{|z|>\varepsilon}\left(\underline{D}_j({K^j}_i-K{\delta^j}_i)+W_{ij}^k({K^j}_k-K{\delta^j}_k)\right)v^idg\nonumber\\
&&+\int_{z=\varepsilon}\left(\frac{1}{2}V^ju+(K^j_{\ i}-K\delta^j_{\ i})v^i\right)n_jd\sigma\nonumber\\
&&\left.-\int_{z=-\varepsilon}\left(\frac{1}{2}V^ju+(K^j_{\ i}-K\delta^j_{\ i})v^i\right)n_jd\sigma\right\}\nonumber\\
&=&\int_\Sigma(H_G u+(M_G)_iv^i)dg\nonumber\\
&&+\int_S\left(\frac{1}{2}V_\Delta^ju+((K_\Delta)^j_{\ i}-(K_\Delta)\delta^j_{\ i})v^i\right)n_jd\sigma,\label{eq:dc-ex}
\end{eqnarray}
where $H_G$ and $M_G$ are the classical energy density and the momentum density on $\Sigma_\pm$ respectively.
Since the classical dominant energy condition is satisfied on $\{z\neq0\}$, the first term in the left-hand side of Eq. (\ref{eq:dc-ex}) is non-negative.
It is easily proven that $V^in_i=-2k_+$ on $\{z>0\}$ and $V^in_i=-2k_-$.
This implies $V_\Delta^in_i=-2k_\Delta$.
So the integrand of the second term in the right-hand side of Eq. (\ref{eq:dc-ex}) is rewritten as $p_0u+p_iv^i$.
Since the surface is causal corner, this is also non-negative.
Therefore $\cpl{H_G}{u}+\cpl{M_G}{v}$ is non-negative and we can say that the data satisfies the dominant energy condition on $\Sigma$.
\end{proof}

\acknowledgments{
This work is based on author's master thesis at Graduate School of Mathematics, Nagoya University, Japan.
The author would like to thank Tetsuya Shiromizu, Keisuke Izumi, Kenta Oishi and Kazuhiro Aoyama for various and useful discussions.
He would also appreciate Chru{\'s}ciel , Lee and LeFloch for useful conversations on technical points.
This work was partially supported by Grand-in-Aid for Scientific Research (A) No.17H01091 from JSPS.
}


\end{document}